\documentclass[conference,letterpaper]{IEEEtran}

\usepackage[utf8]{inputenc} 
\usepackage[T1]{fontenc}
\usepackage{url}              
\usepackage{cite}             

\usepackage[cmex10]{amsmath}  
\usepackage{mleftright}       
\mleftright                   

\usepackage{graphicx}         
\usepackage{booktabs}
\IEEEoverridecommandlockouts
\usepackage{cite}
\usepackage{amsmath,amssymb,amsfonts}
\usepackage{graphicx}
\usepackage{enumerate}
\usepackage{enumitem}
\usepackage{textcomp}
\usepackage{xcolor}
\usepackage{lipsum}
\usepackage{mathtools}
\usepackage{amsthm}
\usepackage{amsmath}
\usepackage{amssymb}
\usepackage{delarray}
\usepackage{bm}
\usepackage{graphicx}
\usepackage{color}
\usepackage{enumitem}
\usepackage{footmisc}
\usepackage{subfig}
\usepackage{graphicx}
\usepackage{cancel}
\usepackage[export]{adjustbox}
\usepackage{comment}
\usepackage{graphicx}
\usepackage{subfig}
\usepackage{caption}
\usepackage{algorithm2e}
\usepackage{pgfplots}
\usepackage{todonotes}
\usepackage{diagbox}
\usepackage{multirow} 

\pgfplotsset{compat=newest}
\pgfplotsset{plot coordinates/math parser=false}
\newlength\figureheight
\newlength\figurewidth

\newtheorem{theorem}{Theorem}
\newtheorem{corollary}{Corollary}

\newtheorem{prop}{Proposition}

\newtheorem{rem}{Remark}

\newcommand{\HHx}[4]{E\left(\N,\ell,k,\qq \right)}

\newcommand{\N}{n}

\newcommand{\prb}{\upsilon}

\makeatletter
\newcommand*{\rom}[1]{\expandafter\@slowromancap\romannumeral #1@}
\makeatother

\newcommand{\qq}[0]{\alpha}

\newcommand{\mat}[1]{{\mathbf{B}}}

\newcommand{\al}{\alpha}

\newcommand{\cM}{\mathsf{M}^{(c)}}
\newcommand{\cx}{\bm{x}}

\newcommand{\cy}{\bm{y}^{(c)}}

\newcommand{\sM}{\mathsf{M}^{(s)}}

\newcommand{\sy}{\bm{y}^{(s)}}

\newcommand{\mem}[2]{m_{#1 \!M + #2}}

\newcommand{\Repr}[2]{\mathcal{R}_{#1,#2}}
\newcommand{\selec}[2]{\mathcal{S}_{b,i}}

\newcommand{\supp}{\mathsf{supp}}
\newcommand{\bad}{\mathcal{D}}
\newcommand{\badm}[1]{\mathcal{B}_{#1}}
\newcommand{\spar}{\rho_{T}}

\newcommand{\nospar}{\infty}

\newcommand{\Fam}[1]{\mathcal{M}_{#1}}

\newcommand{\hh}{h}

\usepackage{centernot}

\renewcommand{\comment}[1]{ }

\newcommand{\mf}{\mu_p}
\newcommand{\mm}{\mu_m}
\newcommand{\lc}{\left\{}
\newcommand{\rc}{\right\}}

\newcommand{\xtra}[2]{\log_{#1}\!\!\big({#2}\big)}

\newcommand{\TI}{T_{\mathsf{I}}}
\newcommand{\TII}{T_{\mathsf{II}}}
\newcommand{\TIIL}{T_{\mathsf{II,L}}}
\newcommand{\TIIsL}{T_{\mathsf{II,sL}}}
\newcommand{\TSNC}{T_{\mathsf{nC,S}} }
\newcommand{\TNSCI}{T_{\mathsf{C,nS,I}} }
\newcommand{\TNSCII}{T_{\mathsf{C,nS,II}} }
\newcommand{\TSCI}{T_{\mathsf{C,S,I}} }

\begin{document}

\title{Sparsity-Constrained Community-Based \\ Group Testing}
\author{
\IEEEauthorblockN{Sarthak Jain, Martina Cardone,  Soheil Mohajer}
University of Minnesota, Minneapolis, \!MN 55455, \!USA,
\!Email: \{jain0122, mcardone, soheil\}@umn.edu
\\
\vspace{-0.9em}
\thanks{This research was supported in part by the U.S. National Science Foundation under Grant CCF-1907785.} 
}
\maketitle

\begin{abstract}
In this work, 
we consider the sparsity-constrained {community-based} group testing problem, where the population follows a community structure. In particular, the community consists of $F$ families, each with $M$ members. {A number} $k_f$ out of the $F$ families are infected, and a family is said to be infected if $k_m$ out of its $M$ members are infected. 
Furthermore, the sparsity constraint allows at most $\spar$ individuals to be grouped in each test. 
For this sparsity-constrained community model, we propose a probabilistic group testing algorithm that can identify the infected population with a vanishing probability of error and we provide an upper-bound on the number of tests.  
When $k_m = \Theta(M)$ and $M \gg \log(FM)$, our bound outperforms the existing sparsity-constrained group testing results trivially applied to the community model. If the sparsity constraint is relaxed, our achievable bound reduces to existing bounds for community-based group testing. 
Moreover, our scheme can also be applied to the classical dilution model, where it outperforms existing noise-level-independent schemes in the literature.
\end{abstract}
\section{Introduction}
Group testing (GT), 
first introduced in 1943 in~\cite{dorfman1943}, is an umbrella term for the methods used to identify $k$ defective items among $n$ items, with as few tests as possible.
{The main idea consists of} performing tests on pools/groups of items rather than testing each item individually. 
GT has many applications, ranging from medicine~\cite{verdun2021} to engineering~\cite{berger1984}, and  
is broadly classified into \emph{combinatorial} GT and \emph{probabilistic} GT. 
In combinatorial GT, the goal is to identify the defective items with a zero error probability~\cite{Du1993}. In probabilistic {GT, instead, it suffices that the error goes to zero asymptotically, as $n \to \infty$;} moreover, for finite $n$, the error can be made arbitrarily small by appropriately scaling the number of tests~\cite{Chan2011,mazumdar2016,barg2017,inan2018, arpino2021, Atia09, cheng2023}. GT can be \emph{noiseless} or \emph{noisy}~\cite{Chan2011, Atia09, cheng2023}. In noiseless GT, each test is error free (no false positives or misdetection); 
whereas, in noisy GT, the {test results} may be erroneous~\cite{scarlett2018,Atia09, CheraghchiIT2011, cheng2023, arpino2021, Chan2011}.

Most GT problems assume a \emph{combinatorial prior} on the set of defective items. This means that the $k$ defective items are equally likely to be any of the $\binom{n}{k}$ items. In this case, the counting bound~\cite{Chan2011} states that the number of tests required to identify the $k$ {defective} items is at least $\Theta\left(k \log\left(\frac{n}{k}\right)\right)$. When $k$ follows a sparse regime, that is, $k = \Theta\left(n^{\delta_k}\right)$ for some constant $\delta_k \in [0,1)$, this is significantly less than individual testing, which requires $\Theta(n)$ tests. For sparse $k$, the counting bound {indeed} becomes $\Theta(k \log(n))$. Several GT schemes achieve the counting bound for noiseless GT with a combinatorial prior and hence, are order optimal~\cite{aldridge2019}.
Recent works have considered variants of GT with additional information on the set of defective items~\cite{cao2023,pmlr2021,karimi2022,pavlos21}. 
In~\cite{pmlr2021}, the authors introduced one such model, referred to as the \emph{community} model, and analyzed the symmetric and general variants of~it. The symmetric community model considers a community of $F$ families, each with $M$ members. {A number} $k_f$ out of the $F$ families are \emph{infected}. If a family is healthy, none of its members are infected; if a family is infected, $k_{m}$ out of its $M$ members are infected. Ignoring the community structure, this model reduces to identifying $k_f k_m$ infected members out of $n=FM$ members, which requires $\Theta\left(k_f k_m \log(n)\right)$ tests. However, it was observed that
leveraging the community structure 
can greatly reduce the number of tests~\cite{pmlr2021,pavlos21}. 

In this work, we consider the symmetric community model of~\cite{pmlr2021}, but we additionally impose a {\em sparsity constraint}. 
This constraint allows at most $\spar$ individuals to be pooled in each test. 
This 
model has practical significance. Many infections, such as {${\text{COVID-19}}$}, are indeed governed by community spread, and a community model is suitable to capture such scenarios. This model {can be helpful also} in bio-security applications, e.g., to test consignments of seeds/flowers~\cite{Clark2023-wr}.  
Moreover, in many real world applications, there is often a constraint on the number of items that can be pooled in each test. {This constraint may depend on several factors, e.g., test equipment capacity and test efficacy.} 
For example, in swab pooling methods for {COVID-19} testing, it is recommended to pool up to $16$ swabs in each test~\cite{Christoff2021-qw}; and some HIV testing schemes allow $80$ individual samples per test~\cite{hiv, gebhard2022}.

For the sparsity-constrained community model, we propose a probabilistic GT scheme that identifies the infected population with a probability greater than $1-n^{-\lambda}$, for any constant {$\lambda > 0$,} and provide an achievable bound on the number of tests required.  When $k_m = \Theta(M)$ and $M \!\gg\! \log(FM)$, our scheme requires much fewer tests than applying existing sparsity-constrained GT schemes~\cite{gandikota2019, gebhard2022} to the community model. {Moreover, without the sparsity constraint, our 
bound reduces to existing bounds in community based GT~\cite{gandikota2019}.
Our scheme can also be applied to the classical dilution model~\cite{arpino2021, CheraghchiIT2011,mazumdar2014group, Atia2009, cheng2023}, where there is no community structure or sparsity constraint. 
For this model, our scheme requires $\Theta\left(\frac{k \log(n)}{\al}\right)$ tests, where $\al$ is the dilution noise parameter. 
{This bound provides a factor of $\alpha$ improvement with respect to the best achievable bound~\cite{jain2022} of $\Theta\left(\frac{k \log(n)}{\al^2}\right)$ in the literature for the dilution model using a noise-level-independent (NLI) scheme~\cite{arpino2021} (i.e., when the test design is independent of $\al$).}


\begin{table}
\caption{Quantities of interest used throughout the paper.}
\label{table:Notation}
\begin{center}
\begin{tabular}{ |c||c| }
\hline
{\bf{Quantity}} & {\bf{Definition}} \\
\hline
\hline
$F$ & Number of families \\
$M$ & Number of members in each family \\
$n$ & Total number of members, that is, $n=FM$ \\
$\bad$ & Set of infected families \\
$k_f$ & Number of infected families \\
$\badm{f}$ & Set of infected members in family $f \in [F]$ \\
$k_m$ & Number of infected members in an infected family\\
$\spar$ & Maximum number of members allowed in each test \\ 
$T$ & Number of tests performed by a GT scheme \\
\hline
\end{tabular}
\end{center}
\vspace{-2em}
\end{table}

{\noindent {\bf{Notation}.} For any $k \in \mathbb{N}$, we define ${[k] := \{1,2, \ldots,k\}}$. For a set $\mathcal{X}$, $|\mathcal{X}|$ denotes its cardinality. For a matrix $\mathsf{M}$, we use $\mathsf{M}_{i,:}$ and $\mathsf{M}_{:,j}$ to represent its $i$th row and $j$th column, respectively. 
An empty set is denoted by $\varnothing$. 
For a vector $\bm{x}$, {we let} ${\supp(\bm{x}) := \{i: \bm{x}_i \neq 0\}}$. Finally, $\wedge$ and $\vee$ represent the Boolean \texttt{AND} and \texttt{OR} operations, respectively. 

\section{System Model} \label{sec:system-model}

We consider a collection of $F$ families, denoted by~$[F]$, where each family consists of~$M$ members.\footnote{{This is the
symmetric} model in~\cite{pmlr2021}, where $M_f = M$, for all $f \in [F]$.} The total number of members is, therefore, ${\N := FM}$. The members of family $f \in [F]$ are referred to as ${\Fam{f}:=\{\mem{(f-1)}{i} : i \in [M]\}}$. An unknown subset $\bad \subseteq [F]$, consisting of $k_f$ families (that is, $|\mathcal{D}|=k_f$), is infected. We assume a combinatorial prior on this subset of infected families, that is, the defective set is chosen uniformly at random among all {the} $\binom{F}{k_f}$ sets of this size $k_f$. If a family $f$ is not infected, none of its members are infected; whereas, if $f$ is infected, an unknown subset ${\badm{f} \subseteq \Fam{f} }$ of the $M$ members of that family are infected. Again, for the symmetric model considered here, we assume that $|\badm{f}| = k_m$ for all $f \in \bad$. Moreover, we assume {that} $\badm{f}$ is chosen uniformly at random among all {the} 
$\binom{M}{k_m}$ subsets of size $k_m$. Table~\ref{table:Notation} summarizes the quantities used for problem formulation. 

{
{Our goal 
is to design a GT scheme, which uses as few tests as possible, to identify the infected population with a vanishing error probability, i.e., a probability of error that goes to zero at a rate of~$n^{-\lambda}$ for some constant $\lambda > 0$.
Due to practical considerations~\cite{Christoff2021-qw, Fernandez2020}, we impose a \emph{sparsity constraint}, which restricts the number of members that can participate in each test. In particular, in any given test, at most $\spar$ out of the $n$ members can be pooled together. }


\section{Preliminaries and Related Results}
In this section, we first introduce the contact matrix, which is the mathematical model for GT. Then, we review some existing results. In particular, we establish two benchmarks for the performance of the algorithm, based on existing methods.



\subsection{Combinatorial GT}
\label{sec:CGT}
Consider a general (with no sparsity constraint or community structure) GT problem with $N$ items among which $k$ are defective. 
Let $T$ be the number of tests performed by a GT algorithm. These tests can be described using the {\em contact} matrix $\cM \in \{0,1\}^{T \times N}$, where each row corresponds to a test and each column corresponds to an item.  If $\cM_{t,i} = 1$, then item $i$ is selected in test~$t$. 
Let $\cx \in \{0,1\}^N$ be the indicator vector for the defective items, that is, $\cx_i = 1$ if and only if item ${i}$ is defective. 
Then, the result of the tests can be represented by a vector  $\cy \in \{0,1\}^{T}$ as 
\begin{equation} \label{eq:ideal_y}
    \cy = \cM \odot \cx,
\end{equation}
where $\odot$ denotes the matrix-vector \emph{logical} multiplication, in which the arithmetic multiplication and addition are replaced by logical \texttt{AND} and \texttt{OR}, respectively. More precisely, we have ${\cy_t = \bigvee_{i=1}^N (\cM_{t,i} \land \cx_i)}$. 
It is known that using a proper selection of 
{$\cM$}
and an appropriate decoder, with probability at least $1-N^{-\lambda}$ {for any $\lambda>0$,} the set of defective items can be identified using $T=\Theta(k\log(N/k))$ tests~\cite{Du1993}. 

\subsection{Sparsity-Constrained Combinatorial GT}
\label{sec:spar-GT}
The result of~\cite{Du1993} holds when the number of items to be tested together in each pool is arbitrary. In general, a larger number of tests is required if a sparsity constraint is imposed~\cite{gandikota2019,gebhard2022,price23}. Let $\rho_U$ be the maximum number of items allowed to participate in each test. From these results and classical GT~\cite{Chan2011}, it can be argued that, to achieve a probability of error of $\widetilde{N}^{-\lambda}$ for some $\widetilde{N} \geq N$ and any constant {$\lambda > 0$,} the number of tests\footnote{The additional term of $\Theta\left({\log_N(\widetilde{N})}\right)$ {in~\eqref{eq:tildeT},} compared to~\cite{gandikota2019, Chan2011}, guarantees {that the error probability vanishes} at the desired rate of $\widetilde{N}^{-\lambda}$ instead of $N^{-\lambda}$.}  required is at least equal to~\cite{Chan2011, gandikota2019},
\begin{align}\label{eq:tildeT}
    \widehat{T}\left(\!N,k,\rho_U, \widetilde{N}\!\right) \!=\! \Theta\left(\max\lc \frac{N}{\rho_U}, k \log(N)\! \rc  \xtra{N}{\widetilde{N}} \!\right).
\end{align} 
 {In our system model (Section~\ref{sec:system-model}),} we can ignore the community structure and directly identify all the $k_f k_m$ infected members out of the $n$ members. With a sparsity constraint $\spar$, the number of required tests can be found from~\eqref{eq:tildeT} as
\begin{align} \label{eq:T-no-comm}
\TSNC = \widehat{T}\left(n, k_f k_m, \spar, n\right)
=\Theta\left(\max\lc\frac{n}{\spar}, k_f k_m \log(n)\hspace{-2pt}\rc\right).
\end{align}

\subsection{{Community-Based} GT Without Sparsity Constraints}
In the system model (Section~\ref{sec:system-model}), if there is no sparsity constraint (that is, {$\spar = \nospar$}), 
a two-stage algorithm, introduced in~\cite{pmlr2021}, can be utilized, where: (i) in the first stage, the $k_f$ infected families are identified; and (ii) in the second stage, depending on the regime of $(k_m,M)$, either individual testing
or GT is performed only on the infected families (identified in the first stage) to identify their $k_m$ infected members. 
{For both 
stages, this} algorithm leverages existing non-adaptive probabilistic GT schemes~\cite{aldridge2019}. The numbers of tests in the first stage and second stage, respectively, are given by
\begin{align}\label{eq:T-nS-C}
\begin{split}
    \TNSCI &= \Theta\left(k_f \log(n)\right),\\
    \TNSCII &= \begin{cases}
        k_f \Theta(M) &\text{ if } k_m = \Theta(M),\\
         k_f \Theta(k_m \log(n)) &\text{ if } k_m = o(M).
    \end{cases} 
\end{split}
\end{align}

\subsection{Incorporating {Sparsity in the Two-Stage Algorithm}}
\label{sec:nS-C}
In the first stage of the algorithm of~\cite{pmlr2021}, initially a  contact matrix is designed to identify the $k_f$ infected families. However, since tests should be applied {on the} individual members (rather than the families), once a family is selected to participate in a test, all of its members will be pooled to be tested. Therefore, since each family consists of $M$ members, in order to satisfy the sparsity constraint of $\spar$ (on the number of {members allowed in each test}), we can pool together at most $\frac{\spar}{M}$ families to be tested. In other words, the initial test matrix should be designed with a sparsity constraint of $\frac{\spar}{M}$. Hence, using~\eqref{eq:tildeT}, the number of tests required in the first stage of the algorithm 
is given by
\begin{align}\label{eq:T-noiseless}
\TSCI &= \widehat{T}\left(F, k_f, \frac{\spar}{M}, n\right)
\nonumber
\\&= \Theta\left(\max\lc\frac{F M}{\spar}, k_f \log(F)\rc \xtra{F}{n} \right).
\end{align}

\section{The Proposed GT Scheme}
\label{sec:OurGTScheme}
In this section, we propose a new sparse GT algorithm to identify the infected members in the community structured problem. 
Inspired by~\cite{pmlr2021} (where there is no sparsity constraint), we adopt a
two-stage GT procedure. 
In the first stage (see Section~\ref{subsec:FirstStage}), the goal is to identify the $k_f$ infected families, whereas in the second stage (see Section~\ref{subsec:SecondStage}), we perform GT only on the infected families (identified in the first stage) to identify their $k_m$ infected members. We denote by $\TI$ and $\TII$  the number of tests required in the first and second stages, respectively. Then, the total number of tests required by the proposed algorithm is given by $T = \TI + \TII$. 


\begin{table}
\caption{Quantities of interest used in the GT scheme.}
\label{table:Notation_scheme}
\begin{center}
\begin{tabular}{ |c||c| }
\hline
{\bf{Quantity}} & {\bf{Definition}} \\
\hline
\hline
$\TI$ & Number of tests in the first stage \\
$\TII$ & Number of tests in the second stage \\
$\rho$ & Number of families selected in each test \\
$r$ & Number of members sampled from each selected family \\
$\al$ & Probability that an infected family is active \\
$\bad_t$ & Set of active infected families during test $t$ \\
$d$ & Threshold for the $d$-threshold decoder \\
$\Repr{f}{t}$ & Members of a selected family $f$ that participate in test $t$ \\ 
\hline
\end{tabular}
\end{center}
\end{table}


\subsection{First Stage: Identifying Infected Families}
\label{subsec:FirstStage}
We use a contact matrix $\cM\in\{0,1\}^{\TI\times F}$, initially designed for $F$ families for the first stage of the algorithm (similar to Section~\ref{sec:CGT} with $(N,k,T)=(F,k_f,\TI)$). For simplicity, we assume that $k_f \geq 2$ and $F \geq 2 k_f$. Table~\ref{table:Notation_scheme} summarizes the parameters used in the proposed scheme.

\noindent {\bf{Probabilistic design of the contact matrix:}} We first choose a $\TI\times F$ contact matrix $\cM$ with a \emph{family-sparsity} parameter $\rho \in [F]$ (which will be determined later). To this end, each row of $\cM$ is uniformly, randomly, and independently from other rows, selected from the $\binom{F}{\rho}$ possible rows that have Hamming weight equal to $\rho$. 
\hfill $\square$

\noindent {\bf{Family representative sets:}} Unlike the scheme in Section~\ref{sec:nS-C}}, where all the members of a selected family participate in a test, we choose a set of \emph{representative} members for each selected family to participate in tests. In particular, for each test $t$, a subset ${\Repr{f}{t}\subseteq \Fam{f} }$ {of members participate} in the test. More formally, the set of individuals that are pooled together in test $t$ is given by $\bigcup_{f \in \supp\left(\cM_{t,:}\right)} \Repr{f}{t}$. To this end, for each $(t,f)$, we select $\Repr{f}{t}$ uniformly at random from all the $\binom{M}{r}$ possible subsets of $\Fam{f}$ of size $r:=|\Repr{f}{t}|=\left\lfloor \frac{\spar}{\rho} \right\rfloor$. With the above designs of $\cM$ and $\Repr{f}{t}$, the number of members {that} participate in each test satisfies
\begin{align}
    \sum_{f=1}^F \cM_{t,f} \; |\Repr{f}{t}| &=\!\!\!\!\!\!\!\!\!\! \sum_{f \in \supp\left(\cM_{t,f}\right)} \! \left\lfloor \frac{\spar}{\rho} \right\rfloor = \rho \left\lfloor \frac{\spar}{\rho} \right\rfloor \leq \spar,
\end{align}
and hence, the sparsity constraint is satisfied.
\hfill $\square$

\noindent {\bf{The sampling matrix:}}
 With the representative sets (instead of the entire family) participating in each test, the identity in~\eqref{eq:ideal_y}} does not hold in general. To see this, consider a case where $\cM_{t,f}=1$, and $\Repr{f}{t} \cap \badm{f} = \varnothing$ for an infected family $f \in \bad$. Then, even {if} $f$ is infected and selected to participate in the test, it will not cause the test $t$ to be positive, since no infected member of the family is in its representative set. In other words, such an infected family \emph{pretends} to be healthy in the test. 
{To capture} this uncertainty, we define a \emph{sampling} matrix ${\sM  \in \{0,1\}^{{\TI} \times F}}$ obtained from 
$\cM$. We call an infected family $f \in \bad$ \emph{active} in test $t$, if and only if, ${\Repr{f}{t} \cap \badm{f} \neq \varnothing}$. 
 We denote the set of active infected families of test $t$ by ${\bad_t \subseteq \bad}$. Then, the sampling matrix $\sM$ is given by
  \begin{equation} \label{eq:sampling1}
\sM_{t,f} = \begin{cases}
   \cM_{t,f}  &  \text{if } f \in \left([F] \setminus \bad\right) \cup \bad_t, \\
   0 & \text{if } f \in \bad \setminus \bad_t,
    \end{cases}    
\end{equation}
%
and the actual results of the  tests (performed on the representatives of the families) are given by 
\begin{align} \label{eq:actual_y}
    \sy = \sM \odot \cx.
\end{align}
To understand the sampling matrix in~\eqref{eq:sampling1}, let us consider an infected family ${f \in \bad}$ that is selected in test $t$ (i.e., ${\cM_{t,f}=1}$). Now, if ${\Repr{f}{t} \cap \badm{f} = \varnothing}$, although $f$ is infected, none of its infected members participate in test $t$. In other words, family $f$ hides its true identity in test $t$. Since $\cM_{t,f}=1$ and $\cx_f=1$, we have {$\cy_{t}=1$}. However, the actual test result $\sy_t$ should not be influenced by $f$. This can be ensured by setting $\sM_{t,f}=0$.  

Let $\al$ be the probability that an infected family is active, that is,
\begin{align} \label{eq:define-al}
\al = \mathbb{P}\left[f \in \bad_t | f \in \bad \right] = 1-\frac{\binom{M-k_m}{r}}{\binom{M}{r}}.
\end{align}
In other words, $\cM_{t,f}=1$ is replaced by $\sM_{t,f}=0$ with probability $1-\alpha$. Moreover, if $\alpha=1$, then $\cM=\sM$. Note that the behavior of $\sM$ and $\cM$ is similar to that of the dilution model, that we recently studied in~\cite{jain2022}, and has also been {investigated} in~\cite{CheraghchiIT2011, atia2012boolean, cheng2023}. 
\hfill 
{$\square$}

}

 Given this construction of $\cM$ and $\Repr{f}{t}$, and the probabilistic nature of $\sM$ and $\sy$, the families are classified as infected or healthy using the following $d$-threshold decoder.

\noindent \textbf{The $d$-threshold decoder: } Let ${\sy_t \in \{0,1\}}$ be the result of test $t\in [\TI]$, given by ${\sy_t = \sM_{t,:} \odot \cx}$.
We define the score $S_{f,t}$ of  family $f$ in test $t$ as 
\begin{align}\label{eq:scoring}
    S_{f,t} = \begin{cases}
        1 &\text{if } \cM_{t,f} =1 \text{ and }  \sy_t=1, \\
        0 & \text{otherwise. }
    \end{cases}
\end{align}
Then, for a given $d>0$, family $f$ is marked as infected 
if and only if
$S_f = \sum\nolimits_{t=1}^{\TI} S_{f,t} \geq d$.
\hfill $\square$

{The following theorem provides the number of tests required in the first stage of the algorithm to ensure that the construction above can decode $\cx$ with an overwhelming probability.}
\begin{theorem} \label{thm1}
There exists a choice of the parameters $\left(\rho,  d \right)$ such that the $d$-threshold decoder requires at most
\begin{subequations}
\begin{align} \label{eq:T(rho)}
\TI =
        {\min_{\rho \in [\widehat{\rho}]}\frac{\zeta (1
        +\lambda) F\log(n)}{\rho \;\! \al} }
        & \leq {\frac{\zeta(1+\lambda) F\log(n)}{f\left(\widehat{\rho}\right)}}
\end{align}
tests
to identify the $k_f$ infected families
with  
error probability ${\mathsf{P}_e \leq n^{-\lambda}}$, for any $\lambda > 0$, where {$\al$ is given in~\eqref{eq:define-al} and}
\begin{align}
f(\rho) &=  \rho \left(1-\left(1-\frac{k_m}{M}\right)^{\frac{\spar}{2\rho}}\right),\ \zeta= 64 \;  {\rm{e}}^4, 
\\ \widehat{\rho} &= \min\left\{\spar, \left\lfloor\frac{F}{2 k_f}\right\rfloor\right\}.
%
\end{align} 

\label{eq:TMainTheorem}
\end{subequations} 
\end{theorem}
\begin{proof}
The proof of Theorem~\ref{thm1} and the choice of the parameters $\left(\rho, d \right)$ (see~\eqref{eq:choice-parameters}) are provided in Section~\ref{sec:thmproof}.
\end{proof}
{
\begin{rem}
    Our scheme is noise-level-independent (NLI)~\cite{arpino2021} because the construction of $\cM$ does not depend on the noise parameter $\al$. 
    With no sparsity constraint, i.e., {$\spar =\nospar$,}
    we have that  $\widehat{\rho} = \left\lfloor\frac{F}{2 k_f}\right\rfloor$.
    {Our proposed scheme can then} be used with the classical dilution model~\cite{Atia09, arpino2021, CheraghchiIT2011, cheng2023, jain2022}, where: (i) the task is to identify $k$ defective items out of $n$ items; and (ii) the defective items exhibit a dilution effect with probability $\al$, {independent of $\rho$. This leads to $\TI = \Theta\left(\frac{k_f \log(n)}{\al}\right)$.} To the best of our knowledge, the best achievable bound in the literature for the dilution model using a NLI GT scheme is $\Theta\left(\frac{k \log(n)}{\al^2}\right)$ tests~\cite{jain2022} and our scheme outperforms this by a factor of $\al$.
\end{rem}
}
\subsection{Second Stage: Identifying All the Infected Members} \label{subsec:SecondStage} 
To identify all the $k_f k_m$ infected members, we can either perform individual testing or sparsity-constrained GT, for each of the $k_f$ families identified in the first stage. 
For the \emph{linear} regime of $k_m$ (i.e., ${k_m = \Theta(M)}$), individual testing (which has sparsity of $1$) is preferred. 
In this case, we would require
\begin{equation} \label{eq:T2-ns}
\TIIL =  k_f\Theta ( M)
\end{equation}
tests. Otherwise, if $k_m$ follows a sub-linear regime (i.e., ${k_m = o(M)}$), performing sparsity-constrained GT (see Section~\ref{sec:spar-GT}) in each of the $k_f$ infected families would be preferred. 
This would require a number of tests equal to
%
\begin{equation}
\label{eq:T2-s}
\TIIsL \!=\! \left \{
\begin{array}{ll}
\!\!\! k_f \Theta\left(\frac{M }{\spar} \frac{\log(n)}{\log(M)}\right) & \text{if } \spar \!=\! o\left(\!\frac{M}{k_m}\!\right),
\\
\!\!\! k_f{\Theta\left( k_m \log(n) 
\right)} 
& \text{if }  \spar \!=\! \Theta\left(\!\frac{M}{k_m}\!\right).
\end{array}
\right .
\end{equation}
Hence, depending on the regime of $M$, the number of tests $\TII$ for the second stage, can be obtained from~\eqref{eq:T2-ns}~or~\eqref{eq:T2-s}. 


\section{Analysis and Comparison} \label{sec:comparison}
In this section, we further analyze the performance (in terms of number of tests required) of the GT scheme {proposed in Section~\ref{sec:OurGTScheme}.
It should be noted that all the comparisons are order-wise, and the multiplicative constants behind the ${\Theta}$ notation are ignored. 
In particular,} from Theorem~\ref{thm1} we have the following corollary.
\begin{corollary}
\label{corollary:AnalysisT}
It holds that 
\begin{align} \label{eq:T-simplified}
\TI \leq 
\Theta \left( \max \lc \frac{FM}{\spar k_m} ,k_f \rc  \log(n)\right ).
\end{align}
\end{corollary}
\begin{proof}
The proof can be found in Appendix~\ref{app:ProofCorollary}.
\end{proof}
%
We now compare the performance of our scheme with existing results. Note that, due to the structure of the problem, the primary interest is {on a specific regime} of parameters, {namely: 
(i) the total number of infected members 
falls within a}
sparse regime, i.e., $k_f k_m = o(n)$ (otherwise individual testing would be optimum); (ii) once a family is infected, a significant number of its members get infected, i.e., $k_m = \Theta(M)$; and {(iii)} the size of the families {is not} very small, i.e., ${M ={\omega}(\log(n))}$ (otherwise each family can be thought 
as an individual). 


\noindent {\bf{$\bullet$ Ignoring the community structure.
}} A naive algorithm that does not exploit the community structure of the problem was discussed in Section~\ref{sec:spar-GT}. For the regime of interest on $(k_m,k_fk_m, M)$, the ratio of the total (both stages) number of tests required {by the two} algorithms can be bounded as 
\begin{align}
    \frac{\TI+\TIIL}{\TSNC} &\leq 
    \frac{\Theta \left( \max \lc \frac{n}{\spar k_m} ,k_f \rc  \log(n) + k_f M\right)}{\Theta\left(\max\lc\frac{n}{\spar}, k_f k_m \log(n)\rc\right)}
    \nonumber
    \\& = \Theta \left( \frac{\log(n)}{M} + \frac{1}{\log(n)}\right ).
    %
    \label{eq:FinalCompBaseline2}
\end{align}
From~\eqref{eq:FinalCompBaseline2}, we note that exploiting the community structure offers an order-wise reduction in the number of tests. 

\noindent {\bf{$\bullet$ Enforcing sparsity for the community-based scheme.}} As discussed in Section~\ref{sec:nS-C}, we can arrive at a sparse GT scheme that leverages the community structure of the problem. 
{This requires} $\TSCI$ tests {(see~\eqref{eq:T-noiseless})} in its first stage, while the number of tests required in the second stage is identical to that of our proposed algorithm (given in~\eqref{eq:T2-ns}~or~\eqref{eq:T2-s}). Since both {schemes} require the same number of tests in the second stage, we only compare their performance in the first stage. We {have that}
\begin{align}
&\frac{\TI}{\TSCI} \leq 
\frac{\Theta \left( \max \lc \frac{n \log(F)}{\spar k_m}, k_f \log(F) \rc \right )}{\Theta \left ( \max \lc \frac{n}{\spar}, k_f \log(F) \rc \right )} \notag
\\ &=
\left \{
\begin{array}{ll}
\Theta \left( \frac{\log F}{M} \right ) & \text{if } 1 \leq \spar < \frac{n}{k_m k_f}, 
\\ \Theta \left( \frac{\spar k_f \log(F)}{n} \right ) & \text{if } \frac{n}{k_f k_m} \leq \spar < \frac{n}{k_f \log(F)},
\\ \Theta(1) & \text{if } \spar \geq \frac{n}{k_f \log(F)}.
\end{array}
\right .
\label{eq:FinalCompBaseline1}
\end{align}
From~\eqref{eq:FinalCompBaseline1}, we note that our scheme outperforms (order-wise) the scheme of~\cite{pmlr2021} for a wide range of parameters. 
When no sparsity constraint is imposed {(i.e., ${\spar = \nospar}$),} $\TSCI = \Theta\left(k_f \log(n)\right)$, which is identical to the number of tests required by our scheme (see~\eqref{eq:T-simplified}). Therefore,  without any sparsity constraint, our scheme performs equivalent to the two-stage scheme of~\cite{pmlr2021}. It is also worth noting that the scheme of {Section~\ref{sec:nS-C}} is not feasible when $\spar < M$.

\section{Proof of Theorem~\ref{thm1}} \label{sec:thmproof}

In this section, we prove Theorem~\ref{thm1}. We use two propositions that are stated next. Specifically, Proposition~\ref{prop:prereq} bounds
\begin{align}
&\mf=\mathbb{E}\left[S_{f} \bigm|  f \notin \bad\right] \ \text{and} \ \mm=\mathbb{E}\left[S_{f} \bigm|  f \in \bad \right],
\end{align}
where $S_f = \sum\nolimits_{t=1}^{\TI } S_{f,t}$ with $S_{f,t}$ defined in~\eqref{eq:scoring}.
{We note that 
$S_f$ of 
$f \in \bad$ is expected to be higher than 
$S_f$ of 
$f \notin \bad$. This is formally shown by Proposition~\ref{prop:prereq}.}
\begin{prop} \label{prop:prereq}
    For $x \in [k_f]$, let $\hh_x$ be defined as
\begin{align}
\label{eq:hx}
    \hh_x := \sum_{\ell=0}^x \binom{x}{\ell} \al^\ell (1-\al)^{x-\ell} \left(1-\frac{ \binom{F-\ell-1}{\rho-1}}{\binom{F}{\rho}} \right),
\end{align}
for $\al$ 
given by~\eqref{eq:define-al}.
Then, for any $\rho$ in the interval $\left[\left\lfloor\frac{F}{2 k_f}\right\rfloor\right]$,
\begin{align}
    &\text{(i) } \hh_x \leq \left(1-\frac{\rho}{F}\right)+\frac{\al \rho}{F}, \nonumber \\
    &\text{(ii) } \mf = \TI \left(\hh_{k_f} - \left(1-\frac{\rho}{F}\right)\right) \leq \TI \frac{\al \rho }{F}, \nonumber \\
    %
    &\text{(iii) } \mm = \TI \left(\al \!+\! (1\!-\! \al) \hh_{k_f-1} \!-\! \left(1-\frac{\rho}{F}\right) \right) 
    \leq  \TI \frac{2 \al  \rho }{F}, \nonumber \\
    &\text{(iv) } \mm - \mf \geq
    \frac{\al \rho \TI}{2F} {\rm{e}}^{-2}. \notag
\end{align}
\end{prop}
\begin{proof}
The proof can be found in Appendix~\ref{proof:prereq}.
\end{proof}
The next proposition will be useful in the proof of Theorem~\ref{thm1} for choosing the family-sparsity parameter $\rho$.
\begin{prop} \label{prop:opti}
    Let $U \in \mathbb{N}$ and $\prb \in (0,1)$. Then,
    \begin{align}
        & \arg\max_{\rho\in [U]} \rho \left(1-\prb^{\frac{\spar}{\rho}}\right) = U. 
    \end{align}
\end{prop}
\begin{proof}
The proof can be found in Appendix~\ref{proof:opti}.
\end{proof}
We are ready to prove Theorem~\ref{thm1}.
Let $\mathsf{P}_{+}$ and $\mathsf{P}_{-}$ be the probabilities of false positive and misdetection errors of the $d$-threshold decoder for a given $f\in [F]$, respectively, i.e.,
\begin{align}
    \mathsf{P}_{+} &= \mathbb{P}\left[S_f \geq d | f \notin \bad \right] \ \text{and} \ \mathsf{P}_{-}  = \mathbb{P}\left[S_f < d | f \in \bad \right].
\end{align}
By the union bound, the total error probability $\mathsf{P}_e$ 
can be upper bounded as
\begin{align} \label{eq:unionbound} 
    \mathsf{P}_e \leq (F-k_f) \mathsf{P}_{+} + k_f \mathsf{P}_{-}.
\end{align}
We choose the following parameters,
\begin{align} \label{eq:choice-parameters}
        \rho = \widehat{\rho},  \ d =  \frac{\mm\! +\!\mf}{2}, \ \TI = \frac{\zeta (1
        +\lambda) F\log(n)}{\rho \al},
    \end{align} 
\noindent where $\widehat{\rho}$ and $\zeta$ are given in~\eqref{eq:TMainTheorem}, $\al$ is given in~\eqref{eq:define-al}, and $\lambda > 0$ is a constant. 
With these choices, we bound $\mathsf{P}_{+}$ and $\mathsf{P}_{-}$ as
\begin{align} 
    \mathrm{P}_+ 
    & = \mathbb{P}\left[S_f \!\geq \! \frac{\mm+\mf}{2} \Bigm| f \notin \bad \right]  \nonumber\\
    %
    &\stackrel{\rm{(a)}}{=} \mathbb{P}\left[S_f \!\geq \! \mf \left(1\!+\!\delta_p \right) \!\Bigm|\! f \notin \bad \right]  \stackrel{\rm{(b)}}{\leq} \exp\left(-\frac{\delta_p^2 \mf}{2+\delta_p}\right) 
     \nonumber\\ 
     & = \exp\left(\!-\frac{(\mm\!-\!\mf)^2}{6 \mf\!+\!2\mm}\right)  \!\stackrel{\rm{(c)}}{\leq}\! \exp\!\left(\!-\frac{\!{\rm{e}}^{-4}  \al \rho \TI}{40 F} \!\right)  \nonumber\\
     &\stackrel{\rm{(d)}}{=} \exp\left(-1.6 (1+\lambda)\log(\N)\right) \leq \N^{-1-\lambda},
     \label{eq:false_pos}
\end{align}
where the labeled (in)equalities follow from:
$\rm{(a)}$ letting ${\delta_p = \frac{\mm-\mf}{2\mf}}\geq 0$;
$\rm{(b)}$ applying Chernoff's bound;
$\rm{(c)}$ using Proposition~\ref{prop:prereq};
and $\rm{(d)}$ using  $\TI$ in~\eqref{eq:choice-parameters}.

The misdetection error {probability can be bounded~as}
\begin{align} \label{eq:false_neg}
    \mathrm{P}_- 
    &\stackrel{\rm{(a)}}{=} \mathbb{P}\left[S_f < \mm \left(1-\delta_m \right) \Bigm| f \in \bad \right]  \nonumber\\
     &\stackrel{\rm{(b)}}{\leq} \exp\left(-\frac{\delta_m^2 \mm}{2}\right) {\stackrel{\rm{(c)}}{\leq} \N^{-1-\lambda}},
\end{align}
where the labeled (in)equalities follow from:
$\rm{(a)}$ {letting} ${\delta_m = \frac{\mm-\mf}{2\mm}} \in (0,0.5]$;
$\rm{(b)}$ using Chernoff's bound;
{and $\rm{(c)}$ using Proposition~\ref{prop:prereq} and $\TI$ in~\eqref{eq:choice-parameters}.}

Combining~\eqref{eq:false_pos} and~\eqref{eq:false_neg} together with the union bound in~\eqref{eq:unionbound}, we get $\mathsf{P}_e\leq n^{-\lambda}$. Furthermore, the number of tests that suffice to achieve this probability of error is given by
\begin{align} 
\TI &  \stackrel{\rm{(a)}}{=} \frac{\zeta (1+\lambda) \frac{F}{\rho}\log(\N)}{ \al} 
\nonumber\\& \stackrel{\rm{(b)}}{=} \frac{\zeta (1\!+\!\lambda) F\log(\N)}{\rho \left(1\!-\!\frac{\binom{M\!-\!k_m}{r}}{\binom{M}{r}}\right)}  \!=\! \frac{\zeta (1\!+\!\lambda) F\log(\N)}{\rho \left(1\!-\!\prod_{j=1}^{r} \left(1 \!-\!\frac{k_m}{M\!-\!j+1}\right)\right)}
\nonumber\\& \leq \frac{\zeta (1+\lambda) F\log(\N)}{\rho \left(1 - \left(1 -\frac{k_m}{M}\right)^r\right)} 
\stackrel{\rm{(c)}}{\leq} \frac{\zeta (1+\lambda) F\log(\N)}{\rho \left(1 - \left(1 -\frac{k_m}{M}\right)^{\frac{\spar}{2\rho}}\right)},
%
\label{eq:tests}
\end{align}
where the labeled (in)equalities follow from:
$\rm{(a)}$ using $\TI$ in~\eqref{eq:choice-parameters};
$\rm{(b)}$ using $\al$ in~\eqref{eq:define-al};
and $\rm{(c)}$ the fact that $\rho \leq \spar$ and hence, $r = \left\lfloor\frac{\spar}{\rho}\right\rfloor \geq \frac{\spar}{2\rho}$. 

To conclude the proof, 
we 
find the value of $\rho$ that minimizes~\eqref{eq:tests}. For this, we analyze the denominator {of the right-hand side} of~\eqref{eq:tests}, which is 
$f(\rho)$ defined in~\eqref{eq:TMainTheorem},
where $\rho \leq \spar$.
In the proof above, we also used Proposition~\ref{prop:prereq}, which requires $\rho \leq \left\lfloor\frac{F}{2 k_f}\right\rfloor$.
Thus, {we need $\rho \leq \widehat{\rho}$, where $\widehat{\rho}$ is defined in~\eqref{eq:TMainTheorem}.}
We therefore seek to maximize $f(\rho)$ (and hence, minimize $\TI$) over the set $\rho \in \left[\widehat{\rho}\right]$.
Substituting $\prb = \left(1-\frac{k_m}{M}\right)^{\frac{1}{2}}$ in Proposition~\ref{prop:opti}, it follows that the optimal choice of $\rho$ is $\rho = \widehat{\rho}$. 
Using $\rho = \widehat{\rho}$ in~\eqref{eq:tests} concludes the proof of Theorem~\ref{thm1}.

\newpage
\appendices

\section{Proof of Corollary~\ref{corollary:AnalysisT}}
\label{app:ProofCorollary}
We start by noting that, since we perform an order-wise analysis, we can make a simplifying assumption that $2 k_f$ divides $F$. 
With reference to~\eqref{eq:TMainTheorem}, this results in
\begin{equation}
\label{eq:RhoT}
\widehat{\rho} =\min\left\{\spar, \left\lfloor\frac{F}{2 k_f}\right\rfloor \right\}= \min\left\{\spar, \frac{F}{2 k_f} \right\}.
\end{equation}
We now analyze two different regimes.

\noindent {\bf{Regime~I} ${\spar \geq\frac{FM}{k_f k_m}}$:} In this case, with reference to~\eqref{eq:RhoT}, we have that $\widehat{\rho} =  \frac{F}{2 k_f}$ and hence, $f(\widehat{\rho})$ in~\eqref{eq:TMainTheorem} can be lower bounded as follows,
\begin{align}
    f(\widehat{\rho}) &= \widehat{\rho} \left(1-\left(1-\frac{k_m}{M}\right)^{\frac{\spar}{2\widehat{\rho}}}\right) \notag  \\
    &\geq \frac{F}{2 k_f} \left(1-\left(1-\frac{k_m}{M}\right)^{\frac{M}{k_m}}\right)  \notag
    \\& \geq  \frac{F}{2 k_f} \left(1-{\rm{e}}^{-1}\right),
    \label{eq:compare_classical}
\end{align}
where the last step follows since $1-x \leq {\rm{e}}^{-x}$ for $x \in [0,1]$.

\noindent {\bf{Regime~II:} $ {1 \leq \spar < \frac{FM}{k_f k_m}}$:} In this case, we have that $f(\widehat{\rho})$ in~\eqref{eq:TMainTheorem} can be lower bounded as follows,
\begin{align}
    f(\widehat{\rho}) &= \widehat{\rho} \left(1-\left(1-\frac{k_m}{M}\right)^{\frac{\spar}{2\widehat{\rho}}}\right) \notag \\
    &\stackrel{\rm{(a)}}{\geq}  \min\left\{\spar, \frac{F}{2 k_f}\right\} \left(1-{\rm{e}}^{-\frac{\spar k_m}{2 \widehat{\rho} M}}\right) \notag \\
    &\stackrel{\rm{(b)}}{\geq} 
        \min\left\{\spar, \frac{F}{2 k_f}\right\} \frac{\spar k_m}{4 \widehat{\rho} M} \notag
        \\& = \frac{F}{2k_f}\min\left\{\frac{2 \spar k_f}{F}, 1\right\}  \frac{\spar k_m}{4 \widehat{\rho} M} \notag \\
        &\stackrel{\rm{(c)}}{=} \frac{F k_m}{8 k_f M}  \min\left\{\frac{2 \spar k_f}{F}, 1\right\} \max\left\{\frac{2 \spar k_f}{F}, 1\right\} \notag
        \\& = \frac{ \spar k_m}{4 M},
        \label{eq:compare_sparsity}
\end{align}
where the labeled (in)equalities follow from:
$\rm{(a)}$ using~\eqref{eq:RhoT} and since $1-x \leq {\rm{e}}^{-x}$ for $x \in [0,1]$;
$\rm{(b)}$ the fact that, from~\eqref{eq:RhoT}, we have that
\begin{equation}
\frac{\spar}{\widehat{\rho}} 
= \max\left\{1,\frac{2 k_f\spar}{F}\right\} \leq \max \left \{ 1, \frac{2M}{k_m} \right \} = \frac{2M}{k_m},
\label{eq:RhoTOverRhoHat}
\end{equation}
and since ${\rm{e}}^{-x} \leq 1-\frac{x}{2}$ for $x \in [0,1]$; 
and $\rm{(c)}$ using the expression of $\spar/\widehat{\rho}$ in~\eqref{eq:RhoTOverRhoHat}.

The proof of Corollary~\ref{corollary:AnalysisT} is concluded by substituting the bounds in~\eqref{eq:compare_classical} and~\eqref{eq:compare_sparsity} inside $\TI$ in Theorem~\ref{thm1}.

\section{Proof of Proposition~\ref{prop:prereq}} \label{proof:prereq}
\subsection{Proof of Proposition~\ref{prop:prereq}(i)}

Note that $\hh_x$ is defined  for ${x \in [k_f]}$. In order to prove the first property, we define
\begin{equation}
v(x,\ell,\al) = \binom{x}{\ell} \al^\ell (1-\al)^{x-\ell}.
\end{equation}
Then, we have  
\begin{align} \label{eq:bound_hx}
    \hh_x & =\sum_{\ell=0}^{x} \binom{x}{\ell} \al^\ell (1-\al)^{x-\ell}   \left(1-\frac{ \binom{F-\ell-1}{\rho-1}}{\binom{F}{\rho}} \right) \notag \\
     &\stackrel{\rm{(a)}}{=} \sum_{\ell=0}^{x} v(x,\ell,\al) \left(  1 \!-\! \frac{\rho}{F\!-\!\ell}\prod_{j=1}^{\ell}\left(1\!-\!\frac{\rho}{F\!-\!j\!+\!1}\right) \right)  \notag \\
&\leq \sum_{\ell=0}^{x} v(x,\ell,\al)\left(1-\frac{\rho}{F} \left(1-\frac{\rho}{F-\ell+1}\right)^\ell \right) \notag \\
    &\stackrel{\rm{(b)}}{\leq}  \sum_{\ell=0}^x v(x,\ell,\al) \left(1-\frac{\rho}{F} \left(1-\frac{2\rho}{F}\right)^\ell \right)  \notag \\
    &=  \sum_{\ell=0}^x \binom{x}{\ell} \al^{\ell} (1-\al)^{x-\ell} \left(1-\frac{\rho}{F} \left(1-\frac{2\rho}{F}\right)^\ell \right) \notag \\
    &=  \sum_{\ell=0}^x \binom{x}{\ell} \al^{\ell} (1-\al)^{x-\ell}  \notag \\
     &\phantom{=} -\frac{\rho}{F} \sum_{\ell=0}^x \binom{x}{\ell} (1-\al)^{x-\ell}  \left(\al-\frac{2\rho\al}{F}\right)^\ell  \notag \\
    & \stackrel{\rm{(c)}}{=} 1-\frac{\rho}{F} \left(1-\frac{2\rho \al}{F}\right)^{x}  \notag \\
    & \stackrel{\rm{(d)}}{\leq} 1-\frac{\rho}{F} \left(1-\frac{2 \al x \rho}{F} \right) \notag \\
    & \stackrel{\rm{(e)}}{\leq} \left(1-\frac{\rho}{F}\right) + \frac{\al \rho}{ F},
\end{align} 
where the labeled (in)equalities follow from:
$\rm{(a)}$ 
expanding the binomial terms;
$\rm{(b)}$ the assumption in 
Proposition~\ref{prop:prereq} that $k_f \leq \frac{F}{2}$, which leads to $ \ell-1 \leq \ell \leq x \leq k_f \leq \frac{F}{2}$;
{$\rm{(c)}$ the Binomial theorem;}
$\rm{(d)}$ the Bernoulli's inequality;
$\rm{(e)}$ the assumption in Proposition~\ref{prop:prereq} that $\rho  \leq \left\lfloor\frac{F}{2 k_f}\right\rfloor \leq \frac{F}{2 k_f}$ {and the fact that $x \leq k_f$.}

\subsection{Proof of Proposition~\ref{prop:prereq}(ii)}
\label{sec:Prop1ii}
We next prove the upper bound on the expected score $\mf$ of a healthy family. In particular, we have that
\begin{align} \label{eq:bound_mf:0}
    \mf & = \mathbb{E}\left[S_{f} \bigm|  f \notin \bad\right] \notag
    \\& \stackrel{\rm{(a)}}{=} \sum_{t=1}^{\TI} \mathbb{E}\left[S_{f,t} \bigm|  f \notin \bad\right] \notag \\
    & \stackrel{\rm{(b)}}{=} \TI \mathbb{E}\left[S_{f,t}  \bigm|  f \notin \bad \right] \notag \\
    &=\TI \mathbb{P}\left[S_{f,t}=1  \bigm|  f \notin \bad \right] \notag  \\
    &\stackrel{\rm{(c)}}{=} \TI \sum_{\ell=0}^{k_f} \mathbb{P}\left[|\bad_{t}|\hspace{-1pt}=\hspace{-1pt}\ell \!\bigm|\! f \notin \bad \right] 
    \hspace{-1pt}\cdot \mathbb{P}\left[S_{f,t}\!=\! 1 \!\bigm| \! f \notin \bad, |\bad_t|\hspace{-1pt}=\hspace{-1pt}\ell \right] \notag \\
    &\stackrel{\rm{(d)}}{=} \TI \sum_{\ell=0}^{k_f} \mathbb{P}\left[|\bad_{t}|=\ell \right] \mathbb{P}\left[S_{f,t}\!=\! 1 \!\!\bigm|\!\!  f \notin \bad, |\bad_t|=\ell \right]  \notag \\
    &\stackrel{\rm{(e)}}{=} \TI \sum_{\ell=0}^{k_f} \mathbb{P}\left[|\bad_{t}|\hspace{-1pt}=\hspace{-1pt}\ell \right] 
    \hspace{-1pt}\cdot \mathbb{P}\left[\cM_{t, f} \hspace{-2pt}=\hspace{-2pt} 1, \sy_t \hspace{-2pt}=\hspace{-2pt}1\!\hspace{-1pt} \bigm| \!\hspace{-1pt} f \hspace{-1pt}\notin\hspace{-1pt} \bad, |\bad_t|\hspace{-1pt}=\hspace{-1pt}\ell \right],
    \end{align}
where the labeled (in)equalities follow from:
$\rm{(a)}$ the linearity of expectation;
$\rm{(b)}$ the fact that, by design, $S_{f,t}$ is identically distributed for all $t \in [\TI]$; $\rm{(c)}$ using the law of total probability;
$\rm{(d)}$ the fact that the number of active infected families $|\bad_t|$ is independent of $f$ being a healthy family and therefore, ${\mathbb{P}\left[|\bad_t|=\ell | f \notin \bad \right] = \mathbb{P}\left[|\bad_t|=\ell\right]}$; 
$\rm{(e)}$ the definition of $S_{f,t}$ in~\eqref{eq:scoring}.  

Then, note that $|\bad_{t}|$ admits a binomial distribution with parameters $(k_f, \alpha)$, leading to 
\begin{align}\label{eq:bound_mf:1}
    \mathbb{P}\left[|\bad_{t}|=\ell \right] = \binom{k_f}{\ell} \al^{\ell} (1-\al)^{k_f-\ell}. 
\end{align}
Moreover, we have 
\begin{align}\label{eq:bound_mf:2}
    \mathbb{P}\Big[\cM_{t, f}& \hspace{-2pt}=\hspace{-2pt} 1, \sy_t \hspace{-2pt}=\hspace{-2pt}1 \!\hspace{-1pt} \bigm| \!\hspace{-1pt} f \hspace{-1pt}\notin\hspace{-1pt} \bad, |\bad_t|\hspace{-1pt}=\hspace{-1pt}\ell \Big] \notag \\
    & = 1- \mathbb{P}\left[\cM_{t, f} = 1, \sy_t =0 \bigm|  f \notin \bad, |\bad_t|=\ell \right]  \nonumber \\
    & \phantom{=1} -\mathbb{P}\left[\cM_{t, f} = 0 \bigm|  f \notin \bad, |\bad_t|=\ell \right]\notag\\
    &= 1-\frac{\binom{F-\ell-1}{\rho-1}}{\binom{F}{\rho}}- \frac{\binom{F-1}{\rho}}{\binom{F}{\rho}}.
\end{align}
Substituting~\eqref{eq:bound_mf:1} and~\eqref{eq:bound_mf:2} into~\eqref{eq:bound_mf:0}, we arrive at 
    \begin{align}\label{eq:bound_mf}
    \mf \!
    &= \TI  \sum_{\ell=0}^{k_f} \binom{k_f}{\ell} \al^{\ell} (1-\al)^{k_f-\ell} 
    \left(1-\frac{\binom{F-\ell-1}{\rho-1}}{\binom{F}{\rho}}- \frac{\binom{F-1}{\rho}}{\binom{F}{\rho}}\right) \notag \\
    &{= \TI  \sum_{\ell=0}^{k_f} \binom{k_f}{\ell} \al^{\ell} (1-\al)^{k_f-\ell} 
    \left(1\!-\!\frac{\binom{F-\ell-1}{\rho-1}}{\binom{F}{\rho}}- \left(1-\frac{\rho}{F}\right)\right)} \notag \\
    &{= \TI 
    \sum_{\ell=0}^{k_f} \binom{k_f}{\ell} \al^{\ell} (1\!-\!\al)^{k_f-\ell} \left(1\!-\!\frac{\binom{F-\ell-1}{\rho-1}}{\binom{F}{\rho}} \right)\hspace{-1pt}-\hspace{-1pt} \TI\left(1\!-\!\frac{\rho}{F}\right)} \notag \\
    & \stackrel{(f)}{=} \TI \left(\hh_{k_f} - \left(1-\frac{\rho}{F}\right)\right) 
     \stackrel{\rm{(g)}}{\leq } \TI\frac{\al\rho}{ F},
\end{align}
where the step labeled by~$\rm{(f)}$ follows from the definition of $\hh_x$ in~Proposition~\ref{prop:prereq}; and the inequality in~$\rm{(g)}$  is a consequence of Proposition~\ref{prop:prereq}(i).

\subsection{Proof of Proposition~\ref{prop:prereq}(iii)}
{Next, we prove the upper bound on the expected {score $\mm$} of an infected family. In particular, we have that
\begin{align} \label{eq:bound_mm}
    \mm &= \mathbb{E}\left[S_{f} \bigm|  f  \in  \bad\right] \notag \\
    & \stackrel{\rm{(a)}}{=}  \sum_{t=1}^{\TI} \mathbb{E}\!\left[S_{f,t} \bigm|  f \in \bad\right] \notag \\
    & \stackrel{\rm{(b)}}{=}  \TI \mathbb{E}\left[S_{f,t} \bigm|  f \in \bad\right]   \notag \\
    & =  \TI \mathbb{P}\left[S_{f,t}=1 \bigm|  f \in \bad\right]  \notag  \\
    & \stackrel{\rm{(c)}}{=} \TI \Big(\mathbb{P}\left[f\in \bad_t| f\in \bad\right] \mathbb{P}\left[S_{f,t}=1 \bigm| \! f \in \bad, f \in \bad_t \right] \nonumber \\ 
    &\;\; +\mathbb{P}\left[f\notin \bad_t | f\in \bad\right] \mathbb{P}\left[S_{f,t}=1  \! \bigm| \! f \in \bad, f \notin \bad_t \right] \Big) \notag  \\
    & \stackrel{\rm{(d)}}{=} \TI \Big(\al \mathbb{P}\left[S_{f,t}=1 \bigm|  f \in \bad, f \in \bad_t \right] \nonumber \\ 
    &\qquad \quad +(1-\al) \mathbb{P}\left[S_{f,t}=1   \bigm|  f \in \bad, f \notin \bad_t \right] \Big)  \notag \\
    & \stackrel{\rm{(e)}}{=} \TI \Big(\al \mathbb{P}\left[\cM_{t,f}=1, \sy_t=1 \bigm| \! f \in \bad, f \in \bad_t \right] \nonumber \\ 
    &\quad +(1\!-\!\al) \mathbb{P}\left[\cM_{t,f}\!=\! 1, \sy_t=1 \! \bigm| \! f \in \bad, f \!\notin\! \bad_t \right] \Big) \notag  \\
    & \stackrel{\rm{(f)}}{=} \TI \Big(\al \mathbb{P}\left[\cM_{t,f}=1 \bigm|  f \in \bad, f \in \bad_t \right] \nonumber \\ 
    &\quad +(1-\al) \left(\hh_{k_f-1} - \left(1-\frac{\rho}{F}\right)\right) \Big) \notag  \\
    & \stackrel{\rm{(g)}}{=} \TI \Big(\al \left(\frac{\rho}{F}\right) \! + \!(1\!-\!\al) \left(\hh_{k_f-1} - \left(1-\frac{\rho}{F}\right)\right) \Big)  \notag \\
    & = \TI \Big(\al + \!(1\!-\!\al) \hh_{k_f-1} - \left(1-\frac{\rho}{F}\right) \Big) \notag\\
    & \stackrel{\rm{(h)}}{\leq} \TI \Big(\al + \!(1\!-\!\al) \left(1-\frac{\rho}{F}+\frac{\al \rho }{ F}\right) \!-\! \left(1-\frac{\rho}{F}\right) \Big)\notag \\
    &=  \TI\left(\frac{\al \rho}{F} + \frac{(1-\al) \al \rho }{F}\right) \notag
    \\ & \leq  \TI\frac{2 \al \rho}{F},
\end{align}
where the labeled (in)equalities follow from:
$\rm{(a)}$ the linearity of expectation;
$\rm{(b)}$ the fact that, by design, $S_{f,t}$ is identically distributed for all $t \in [\TI]$; 
$\rm{(c)}$ using the law of total probability; 
$\rm{(d)}$ the fact that each infected family is active with probability $\al$ (see~\eqref{eq:define-al}); 
$\rm{(e)}$ the definition of $S_{f,t}$ in~\eqref{eq:scoring};  
$\rm{(f)}$ the two facts that: (1) for an active infected family $f \in \bad \cap \bad_t$, $\cM_{t,f}=1$ implies $\sy_t = 1$,  
and (2) if $f$ is an infected but inactive family, it behaves like a healthy family and there are $k_f-1$ potentially active infected families left in the system; hence, we can follow similar computations as in Section~\ref{sec:Prop1ii}, where $k_f$ is now replaced by $k_f-1$;
$\rm{(g)}$ the fact that each family is selected with probability $1-\frac{\binom{F-1}{\rho}}{\rho} = \frac{\rho}{F}$; 
and $\rm{(h)}$ Proposition~\ref{prop:prereq}(i).

\subsection{Proof of Proposition~\ref{prop:prereq}(iv)}
Finally, we prove the lower bound on the difference {of the} expected scores $\mm-\mf$, between an infected and a healthy family. In particular, we have that
\begin{align}\label{eq:prob1-4:0}
    &\frac{\mm-\mf}{\TI } \stackrel{\rm{(a)}}{=}  \al - (\hh_{k_f} - (1-\al) \hh_{k_f-1}) \notag \\
    &\stackrel{\rm{(b)}}{=} \al - \! \sum_{\ell=0}^{k_f} \Bigg[\left(\!\!\binom{k_f}{\ell}\!-\!\binom{k_f\!-\!1}{\ell}\!\!\right) \nonumber \\
    & \qquad \qquad \quad \times \al^{\ell} (1\!-\!\al)^{k_f-\ell} \left(\!1\!-\!\frac{ \binom{F-\ell-1}{\rho-1}}{\binom{F}{\rho}} \!\right)\Bigg] \notag\\
    &= \al-\! \sum_{\ell=0}^{k_f} \!\binom{k_f}{\ell} \al^{\ell} (1\!-\!\al)^{k_f-\ell} \frac{\ell}{k_f}\left(\!1\!-\!\frac{ \binom{F-\ell-1}{\rho-1}}{\binom{F}{\rho}} \!\right) \notag\\
    &= \al- \sum_{\ell=0}^{k_f} \Bigg[\binom{k_f}{\ell} \al^{\ell} (1-\al)^{k_f-\ell} \nonumber \\
    & \qquad \qquad \qquad \times \frac{\ell}{k_f}\Bigg(1-\frac{\rho}{F -\ell}\prod_{j=1}^{\ell} \left( 1\!-\!\frac{\rho}{F-j+1}\right)\Bigg) \Bigg]\notag\\
    &\stackrel{\rm{(c)}}{\geq} \al \!-\!  \al \sum_{\ell=0}^{k_f}\! \binom{k_f}{\ell} \al^{\ell-1} (1-\al)^{k_f-\ell} 
    \frac{\ell}{k_f}\left(\! 1
\!-\frac{\rho}{F}\left(\! 1 \!-\frac{2 \rho}{F}\right)^{\ell} \right), 
    \end{align}
where in $\rm{(a)}$ we substituted $\mf = {\TI} \left( \hh_{k_f} - (1-\frac{\rho}{F}) \right )$ and ${\mm = {\TI} \left(\al + (1-\al)\hh_{k_f-1} - (1-\frac{\rho}{F}) \right )}$ (as obtained in Proposition~\ref{prop:prereq}(ii) and Proposition~\ref{prop:prereq}(iii), respectively); the equality in~$\rm{(b)}$ holds by substituting the expression of $\hh_x$ in~\eqref{eq:hx}; 
and step~$\rm{(c)}$ follows from the assumption $F\geq 2 k_f$.
    
Now, using the identity $\binom{a}{b} = \frac{a}{b}\binom{a-1}{b-1}$, we can continue from~\eqref{eq:prob1-4:0} as follows: 
\begin{align}
&\frac{\mm-\mf}{\TI } \notag\\
&\geq \hspace{-1pt} \al\hspace{-1pt}-\hspace{-1pt} \al \sum_{\ell=1}^{k_f} \!\binom{k_f-1}{\ell-1} \al^{\ell-1} (1\hspace{-1pt}-\hspace{-1pt}\al)^{k_f-\ell} \left(1-\frac{\rho}{F}\left(\hspace{-1pt}1\hspace{-1pt}-\hspace{-1pt}\frac{2 \rho}{F}\right)^{\ell} \right) \notag \\
    &= \al\!-\! \al \hspace{-1pt}\sum_{\ell=0}^{k_f-1}\hspace{-2pt} \binom{k_f\hspace{-1pt}-\hspace{-1pt}1}{\ell} \al^{\ell} (1\!-\!\al)^{k_f-1-\ell} \left(\!1\!-\!\frac{\rho}{F}\left(1\!-\!\frac{2 \rho}{F}\right)^{\hspace{-1pt}\ell\hspace{-1pt}+\hspace{-1pt}1} \hspace{-1pt}\right)\notag \\
    &\stackrel{\rm{(d)}}{\geq} \hspace{-2pt}\al\!-\! \al \!\sum_{\ell=0}^{k_f-1}\hspace{-2pt} \binom{k_f\hspace{-1pt}-\hspace{-1pt}1}{\ell} \al^{\ell} (1\!-\!\al)^{k_f\hspace{-1pt}-\hspace{-1pt}1\hspace{-1pt}-\hspace{-1pt}\ell} \left(\!1\!-\!\frac{\rho}{2F}\left(1-\frac{2 \rho}{F}\right)^{\hspace{-1pt}\ell} \hspace{0pt}\right) \notag\\
    &\stackrel{\rm{(e)}}{=} \al- \al \left(1-\frac{\rho}{2F}\left(1-\frac{2 \al \rho}{F}\right)^{k_f-1} \right) \notag\\
    &\geq \frac{\al \rho}{2F}\left(1-\frac{2 \al \rho}{F}\right)^{k_f} \notag \\
    &\stackrel{\rm{(f)}}{\geq} \frac{\al \rho}{2F}\exp\left(-\frac{\frac{2 \al k_f \rho}{F}}{1-\frac{2 \al \rho}{F}}\right) \notag\\
    &\stackrel{\rm{(g)}}{\geq} \frac{\al \rho}{2F}\exp\left(-\frac{\al}{1-\frac{\al }{k_f}}\right)\notag \\
    &\stackrel{\rm{(h)}}{\geq}
    \frac{\al \rho}{2F}  {\rm{e}}^{-2},
\end{align}
where the labeled (in)equalities follow from: 
{$\rm{(d)}$ the assumption in Proposition~\ref{prop:prereq} that $\rho  \leq \left\lfloor\frac{F}{2 k_f}\right\rfloor \leq \frac{F}{2 k_f}$ and the fact that $k_f \geq 2$;
$\rm{(e)}$ using the binomial theorem;}
$\rm{(f)}$ the inequality $ {{\rm{e}}^{-\frac{x}{1-x}} \leq 1-x \leq {\rm{e}}^{-x}}$, which is valid for any $x \in [0,1]$; 
$\rm{(g)}$ the assumption in Proposition~\ref{prop:prereq} that $\rho  \leq \left\lfloor\frac{F}{2 k_f}\right\rfloor \leq \frac{F}{2 k_f}$; 
and $\rm{(h)}$ the fact that $\al \leq 1$ and the assumption that $k_f \geq 2$.

\section{Proof of Proposition~\ref{prop:opti}} \label{proof:opti}
We let
\begin{equation}
g(\rho) = \rho \left(1-\prb^{\frac{\spar}{\rho}}\right),
\end{equation}
and we show that it is a non-decreasing function of $\rho$.
We have that
\begin{align}
 &\frac{{\rm{d}}{g(\rho)}}{{\rm{d}}\rho} = -\frac{\spar}{\rho} \prb^{\frac{\spar}{\rho}} \log\left(\frac{1}{\prb}\right)  + \left(1-\prb^{\frac{\spar}{\rho}}\right) \notag \\
 &= -\frac{\spar}{\rho} \prb^{\frac{\spar}{\rho}} \log\left(1+\frac{1-\prb}{\prb}\right)  + \left(1-\prb^{\frac{\spar}{\rho}}\right) \notag \\
 & \stackrel{\rm{(a)}}{\geq} -\frac{\spar}{\rho} \prb^{\frac{\spar}{\rho}}  \frac{1-\prb}{\prb} +  \left(1-\prb^{\frac{\spar}{\rho}}\right)  \notag \\
 &=  \frac{\prb^{\frac{\spar}{\rho}}}{\rho} \left(\rho \prb^{-\frac{\spar}{\rho}}-\frac{\spar (1-\prb)}{\prb} -\rho \right) \notag  \\
  &=  \frac{\prb^{\frac{\spar}{\rho}}}{\rho} \left(\rho (1-(1-\prb))^{-\frac{\spar}{\rho}}-\frac{\spar (1-\prb)}{\prb} \!-\!\rho \right) \notag \\
 &\stackrel{\rm{(b)}}{=}  \frac{\prb^{\frac{\spar}{\rho}}}{\rho} \left(\rho \left(1+\sum_{j=1}^{\infty}\left(1\!-\!\prb\right)^j\right)^{\frac{\spar}{\rho}} \!\!\!\!-\frac{\spar (1-\prb)}{\prb} \!-\!\rho \right) \notag \\
 &\stackrel{\rm{(c)}}{\geq}  \frac{\prb^{\frac{\spar}{\rho}}}{\rho} \left(\rho \left(\! 1\!+\!\frac{\spar}{\rho}\sum_{j=1}^{\infty}\left(1\!-\!\prb\right)^j\right)\!-\!\frac{\spar (1-\prb)}{\prb} \!-\!\rho\right) \notag \\
 & =  \frac{\prb^{\frac{\spar}{\rho}}}{\rho} \left(\spar \sum_{j=1}^{\infty} \left(1-\prb\right)^j-\frac{\spar(1-\prb)}{\prb} \right) = 0, 
\end{align}
where the labeled (in)equalities follow from:
$\rm{(a)}$ using the inequality $\log(1+x) \leq x$ for $x > -1$;
$\rm{(b)}$ the fact that $(1-x)^{-1} = \sum_{j=0}^\infty x^j$ for $x \in (0,1)$; 
and $\rm{(c)}$ using the Bernoulli's inequality. 
Since $g(\rho)$ is non-decreasing in $\rho$, then the maximum occurs at $\rho = U$.
This concludes the proof of Proposition~\ref{prop:opti}.

\newpage 

\end{document}